\documentclass[a4paper,11pt]{iopart}

\usepackage{graphicx}
\usepackage{amsfonts,amssymb,amsthm,amsbsy}
\usepackage{fullpage}
\usepackage{graphicx}
\usepackage{xcolor}
\usepackage{cite}

\newtheorem{thm}{Theorem}[section]
\newtheorem{lem}[thm]{Lemma}
\newtheorem{dfn}{Definition}[section]
\newtheorem{cor}[thm]{Corollary}
\newtheorem{remark}{Remark}[section]

\newcommand{\R}{\mathbb{R}}
\newcommand{\calS}{\mathcal{S}}
\newcommand{\E}{\mathcal{E}}
\newcommand{\V}{\mathcal{V}}
\newcommand{\F}{\mathcal{F}}
\newcommand{\Z}{\mathcal{Z}}

\newcommand{\f}{\boldsymbol{f}}
\newcommand{\g}{\boldsymbol{g}}

\renewcommand{\vec}[1]{\boldsymbol{#1}}

\newcommand{\Iprod}[2]{\langle #1,#2\rangle}
\newcommand{\sgn}{{\rm sgn}}

\begin{document}

\maketitle

\title{Isospectral discrete and quantum graphs with the same flip counts and nodal counts}
\author{Jonas S. Juul$^{1}$ and Christopher H. Joyner$^{2}$}
\address {$^1$Niels Bohr Institute, University of Copenhagen, Blegdamsvej 17, Copenhagen 2100-DK, Denmark}
\address {$^2$School of Mathematical Sciences, Queen Mary University of London, London, E1 4NS, United Kingdom}
\ead{\mailto{jonas.juul@nbi.ku.dk}, \mailto{c.joyner@qmul.ac.uk} }
\date{\today}

\begin{abstract}
The existence of non-isomorphic graphs which share the same Laplace spectrum (to be referred to as isospectral graphs) leads naturally to the following question: What additional information is required in order to resolve isospectral graphs? It was suggested by Band, Shapira and Smilansky that this might be achieved by either counting the number of nodal domains or the number of times the eigenfunctions change sign (the so-called flip count) \cite{band_nodal_2006,Band-2007}. Recently examples of (discrete) isospectral graphs with the same flip count and nodal count have been constructed by K. Ammann by utilising Godsil-McKay switching \cite{Ammann-2015}. Here we provide a simple alternative mechanism that produces systematic examples of both discrete and quantum isospectral graphs with the same flip and nodal counts.
\end{abstract}

\section{Introduction}
One of the traditional means for studying the structure of eigenfunctions of the Laplacian on a large variety of domains is the counting of their nodal domains - the connected subsets where the eigenfunction has a constant sign. The systematic study begins by ordering the spectrum in a monotonically non-decreasing sequence, and associating with the $n^{\rm th}$ eigenfunction the corresponding nodal domain number $\nu_n$. For domains in one dimension, Sturm's oscillation theorem states that $\nu_n=n$. For higher dimensions Courant's theorem provides the upper bound $\nu_n \le n$ \cite{courant_methods_1953}. Studies of nodal domains have yielded many new and surprising insights into various branches of Physics and Mathematics (see e.g. \cite {gnutzmann_complex_2013} for a collection of relevant papers). In particular, it became apparent that the nodal sequence $\{\nu_n\}_{n=1}^{\infty}$ stores information about the domain on which the Laplacian is defined, such as its boundaries or metric, which does not overlap with the information stored in the spectrum \cite{blum_nodal_2002}. Therefore, since the answer to Kac's famous question \cite{kac_can_1966} `can one hear the shape of a drum? was shown to be `no' (examples of pairs of isospectral\footnote{The term `cospectral' is used in place of `isospectral' by some other authors.} planar domains were first obtained by Gordon Webb and Wolpert \cite{Gordon-1992}), this led Gnutzmann, Smilansky and others \cite{gnutzmann_resolving_2005,gnutzmann_can_2006} to ask `can one \emph{count} the shape of a drum?'. More specfically, given two isospectral domains or manifolds, is the knowledge of the nodal count enough to distinguish the two? Br\"{u}ning and Fajman \cite{bruning_nodal_2012} showed this is not the case for certain flat tori, however it seems to be positive in some particular classes of domains \cite{karageorge_counting_2008}.

In the present article we address this question in the context of both discrete and quantum graphs. Studying the Laplacian in these simple, yet non-trivial, systems has proven to be remarkably insightful because many properties are shared with their higher-dimensional counterparts. For instance, it is known that for tree-graphs the Sturm result holds \cite{berkolaiko_lower_2008}, and more generally Courant's theorem also applies in this context \cite{Davies_2001}. In addition, recent results connecting the stability of eigenvalues under small perturbations to nodal quantities were first born out of discoveries in graphs \cite{berkolaiko_stability_2012,Band-2012,berkolaiko_nodal_2013,colin_de_verdiere_magnetic_2013}, before their application to planar domains \cite{Berkolaiko-2012a}. Counting nodal domains in chaotic billiards \cite{blum_nodal_2002,Bogomolny_2002} is also highly related to analogous studies in $d-$regular graphs, which were carried out numerically \cite{charikar_eigenvectors_2007} and analyzed theoretically \cite{elon_eigenvectors_2008} within a random wave model. We also point the reader to the review article \cite{rami_band_nodal_2008} which provides further motivations and developments in the subject.

In \cite{band_nodal_2006,Band-2007} a number of examples of pairs of isospectral quantum graphs were constructed, which were analogous of previous isospectral domains in $\R^2$  \cite{Gordon-1992,Buser-1994,Chapman-1995,Okada-2001,Jakobson-2006} obtained using Sunada's method \cite{Sunada-1985}. In all of the examples the nodal count was able to distinguish these pairs, leading the authors to conjecture that this could resolve isospectrality. Note that in the case of quantum graphs it was shown that knowing the spectrum is enough to distinguish between quantum graphs, provided the lengths of the edges are all incommensurate \cite{Gutkin-2001}. In the context of discrete graphs it was shown that if the so-called \emph{weighted Laplacian} was used then one can find an example of isospectral graphs that are \emph{not} distinguished by the nodal count \cite{oren_isospectral_2012} but the same could not be said of the standard Laplacian. Recently K. Ammann has used the method of Godsil-McKay switching \cite{Godsil_1982} to provide examples of discrete graphs in which both the flip count and nodal count of the Laplacian is the same \cite{Ammann-2015}. Inspired by her results we show here that one can use an alternative simple mechanism for constructing pairs of both discrete and quantum graphs that are isospectral that have the same flip and nodal counts.

The article is presented as follows: In the remainder of the introduction we recount the necessary properties of both discrete graphs, quantum graphs and their respective nodal and flip counts. In Section \ref{Sec: Discrete graphs} we explain a simple mechanism for obtaining isospectral graphs and then go on to show that from this the flip counts and nodal counts will be the same. We also show in Subsection \ref{Sec: Removing isospectrality} that one can find non-isospectral examples of discrete graphs for which the flip and nodal counts coincide. In Section \ref{Sec: Quantum graphs} we adapt the mechanism to the quantum graph setting and give analogous examples. Finally in Section \ref{Sec: Conclusions} we offer some concluding remarks and possible further directions.

\subsection{Discrete graphs}\label{Sec: Discrete setup}

A discrete graph $G=(\mathcal{V},\mathcal{E})$ is given by a set of vertices $\V$ and (undirected) edges $\E$, meaning $(i,j) = (j,i) \in \E$ if the vertices $i$ and $j$ are connected (we also use the notation $i \sim j$ to denote that $i$ is connected to $j$). The number of vertices and edges are denoted $V = |\V|$ and $E = |\E|$ respectively. In the present context (unless otherwise stated) we assume $G$ to be connected and simple, meaning there are no parallel edges or self loops (i.e. edges of the form $e=(i,i)$). The connectivity of $G$ is encoded in the $V\times V$ \emph{adjacency matrix} $A(G)$ whose $(i,j)$ and $(j,i)$ entries are $1$ if $(i,j)\in\E$ and $0$ otherwise. The degree $d_i$ of a vertex $i$ is the number of vertices that are connected to it. We can extract this quantity from $A$ using $d_i=\sum_j A_{ij}$ and in turn use this to construct the diagonal \emph{degree matrix} $D(G)={\rm diag}(d_1,\ldots,d_V)$. Combining these two matrices forms the discrete Laplacian $L(G): \R^V \to \R^V$ given by
\begin{equation}
L(G)=D(G)-A(G).
\end{equation}
The Laplacian is real, symmetric and positive-definite. It therefore has $V$ non-negative eigenvalues $\lambda_1 \leq \lambda_2 \leq \ldots \leq \lambda_V$ and associated real eigenvectors $\f_n = (f_n(i),\cdots,f_n(V))^T$.

In the following we shall investigate the nodal properties. However eigenvectors that possess zeros present difficulties that must be avoided. It is therefore important to restrict the investigation to eigenvectors of the following type.

\begin{dfn}[Genericity in discrete graphs]\label{Dfn: Generic discrete} Let $\f$ be an eigenvector of $L$ with eigenvalue $\lambda$. Then $\f$ is called \textbf{generic} if the eigenvalue $\lambda$ is non-degenerate and $f(i) \neq 0$ for all $i \in \V$.
\end{dfn}

\subsection{Quantum graphs}
Here we give a brief overview of the construction of quantum graphs. One should consult \cite{Berk-Kuch-2013,Gnutzmann_2006} for a more detailed exposition.

A metric graph is obtained from a discrete graph $G = (\V,\E)$ by endowing each edge $e = (i,j) \in \E$ with a finite length $l_e$. In this sense we can allow for a position $x$ on the graph given by selecting an edge $e$ and writing $x_e \in [0,l_e]$ as the distance from the origin vertex $i = o(e)$. We will denote by $\Gamma = (G,l_1,\ldots,l_{E})$ the corresponding metric graph. The establishment of a metric means we may construct functions $\psi(x)$ that are supported on the graph. We do this by restricting the function to each edge and we denote this by $\psi^e(x_e)$, where $x_e \in [0,l_e]$. 

Given this metric graph we may then consider operators acting on appropriate spaces of functions that are supported on the graph. In our present context we shall choose the free one-dimensional Schr\"{o}dinger operator (the negative Laplacian) acting on each edge
\[
L : \psi^e(x_e) \mapsto - \frac{d^2\psi^e}{dx_e^2}(x_e)
\]
and the correct function space (see e.g. \cite{Berk-Kuch-2013} for details) is given by
\[
H^2(\Gamma) := \bigoplus_{e\in \E} H^2([0,l_e]),
\]
where $H^2([0,l_e])$ denotes the Sobolev space on the edge $e$ of real one-dimensional functions $\psi \in L^2([0,l_e])$ whose weak-derivatives up to order two are square integrable.

The final step is to make the operator self-adjoint with respect to the following inner product 
\begin{equation}\label{Eq: Inner product}
\langle \psi, \phi \rangle := \sum_{e \in \E} \int_0^{l_e} \psi(x)\phi(x) dx.
\end{equation}
This requires introducing vertex conditions that match the function and its derivative on each edge emanating from a vertex. A complete classification of conditions has been obtained \cite{Kost-Schr-1999}, however here we will restrict ourselves to so-called Neumann (also referred to as Kirchoff) vertex conditions, given by the following two properties.

\begin{enumerate}
\item \textbf{The function $\psi$ is continuous at vertices:} For any two vertices $v_i,v_j$ such that $v_i,v_j \sim u$ we have
\begin{equation}\label{Eqn: Neumann cond 1}
\psi^{(u,v_i)}(0) = \psi^{(u,v_j)}(0).
\end{equation}
\item \textbf{The sum of outgoing derivatives is zero at vertices:} For vertices $v_1,v_2,\ldots, v_{d_u}$ such that $v_i \sim u$ we have
\begin{equation}\label{Eqn: Neumann cond 2}
\sum_{i=1}^{d_u} \frac{d\psi^{(u,v_i)}}{dx_{(u,v_i)}} (0) = 0.
\end{equation}
\end{enumerate}
The function $\psi(x)$ is an eigenfunction of $\Gamma$ (with eigenvalue $\lambda$) if $L\psi = \lambda \psi$ and the conditions (\ref{Eqn: Neumann cond 1}) and (\ref{Eqn: Neumann cond 2}) are satisfied. Moreover, since we shall only consider compact graphs, the spectrum of the operator $L$ is a countable sequence of non-negative real numbers with no accumulation points. Thus we can number the eigenvalues $\lambda_1 \leq \lambda_2 \leq \ldots$ with corresponding eigenfunctions $\psi_n(x)$, which form a basis for $L^2(\Gamma)$ - the space of square integrable functions with respect to the inner product (\ref{Eq: Inner product}) and satisfying the boundary conditions (\ref{Eqn: Neumann cond 1}) and (\ref{Eqn: Neumann cond 2}). Again, we say the eigenvalue $\lambda$ is simple if it is non-degenerate.

In the quantum graphs setting there is an analogous version on Definition \ref{Dfn: Generic discrete} for generic eigenfunctions, which shall be needed in order to avoid ambiguity over nodal properties.
 
\begin{dfn}[Genericity in quantum graphs]
An eigenfunction $\psi$ of $L_{\Gamma}$ is called \emph{generic} if it is non-zero on all vertices and the corresponding eigenvalue $\lambda_n$ is simple.
\end{dfn}

\subsection{The flip count and nodal count}

The definition of generic eigenvectors (in the discrete case) and generic eigenfunctions (in the quantum case) allows us to proceed, without ambiguity, to defining the flip count and nodal (domain) count in each scenario.\footnote{Note that a number of other authors refer to our definition of the `flip count' as the `nodal (point) count', e.g. \cite{berkolaiko_nodal_2013,band_nodal_2013}. We prefer this terminology in this context to emphasise the distinction between flips and nodal domains.}

\begin{dfn}[Discrete graph flip count]Let $\f$ be a generic eigenvector of $G$. Then 
\[
\F_G(\f) := \{(i,j) \in \E(G) : f_n(i)f_n(j) <0 \}
\]
denotes the set of (undirected) edges on the discrete graph $G$ for which $\f$ changes sign (i.e. flips). $\mu_G(\f) :=  |\F_G(\f)|$ is then the number of flips of $\f$.
\end{dfn}

\begin{dfn}[Discrete graph nodal count]Let $\tilde{G} = (\V(G),\E(G)\setminus \F(\f))$ be the graph formed by removing all edges on which a generic eigenvector $\f$ changes sign. The number of nodal domains of $\f$ on $G$ is denoted $\nu_G(\f)$ and defined to be the number of connected components of $\tilde{G}$.
\end{dfn}

In the case of quantum graphs we have a completely analogous formalism. Note that due to the vertex conditions (\ref{Eqn: Neumann cond 1}) and (\ref{Eqn: Neumann cond 2}) we are guaranteed that the zeros of generic eigenfunctions $\psi_n$ are points on the edges of the graphs.

\begin{dfn}[Quantum graph flip count]
Let $\psi$ be a generic eigenfunction of $L_\Gamma$ and $\Z_{\Gamma}(\psi) := \{x \in \Gamma : \psi(x) = 0\}$ the associated zero set. Then the flip count is the number of times the eigenfunction changes sign on edges, or
\begin{equation}
\mu_{\Gamma}(\psi) =  \sum_{e \in \E} \int_0^{l_e} \delta_{\Z_{\Gamma}(\psi)}(x) dx .
\end{equation}
\end{dfn}
\begin{dfn}[Quantum graph nodal count]
Let $\Gamma \setminus \Z_{\Gamma}(\psi)$ be the graph obtained by removing the zero set of $\psi$. Then the number of nodal domains $\nu_{\Gamma}(\psi)$ is the number of connected components of $\Gamma \setminus \Z_{\Gamma}(\psi)$.
\end{dfn}

In both discrete and quantum graphs the number of nodal domain $\nu_n = \nu_G(\f_n)$ (resp. $\nu_n = \nu_{\Gamma}(\psi_n)$) of the $n^{\rm th}$ eigenvector (resp. eigenfunction) are known to obey certain general bounds. An upper bound is given by Courant's theorem $\nu_n \leq n$, which was translated from the setting of domains to discrete graphs in \cite{Davies_2001} and to quantum graphs in \cite{Gnut-Smi-Web-2003}. A lower bound was established for both discrete and quantum graphs by Berkolaiko \cite{berkolaiko_lower_2008}, so that for generic eigenvectors we have
\begin{equation}
\label{Nodal bounds}
n-\beta \le \nu_n\le n\ .
\end{equation}
Here $\beta = E - V +1$ is the first Betti number, which counts the number of fundamental cycles, or the minimum number of edges one must remove in order for the graph to become a tree. For the flip count we have
\begin{equation}
\label{Flip bounds}
n-1\le \mu_n\le n-1 +\beta \ .
\end{equation}
The lower bound was also proved in \cite{berkolaiko_lower_2008}, the upper bound follows from the fact that $\mu_n \leq \nu_n - 1 + \beta$. The two bounds imply immediately that for trees (\ref{Nodal bounds}) and (\ref{Flip bounds}) reduce to the Sturm oscillation relation $\nu_n=\mu_n+1=n$, which means tree graphs cannot be distinguished via there flip or nodal counts. It has also been proved, conversely, that all graphs with flip counts that follow the integer sequence from $1$ to $V$ must be trees \cite{band_nodal_2013}. However, interestingly, the analogous statement about the nodal count remains illusive.

We would also like to highlight that, although the flip count and nodal count are obviously related (see \cite{rami_band_nodal_2008} for further discussion on this relation), the two contain different types of information. The flip count is \emph{local}, as the sign changes of eigenvectors or eigenfunctions occur across small distances, however the nodal count is truly a \emph{global} quantity, as nodal domains can stretch across significant proportions of the graph. For this reason the flip count is often a much easier quantity to obtain, both numerically and analytically \cite{rami_band_nodal_2008}, than the nodal count.

\section{Discrete graphs}\label{Sec: Discrete graphs}

In this section we present a mechanism for constructing both isospectral and non-isospectral pairs of graphs with the same flip count and nodal count. It is based upon the idea of inserting edges between dangling bonds, or leaves, of a graph. To begin we introduce the following definition.
\begin{dfn} A \textbf{leaf} is a vertex $i$ such that $d_i = 1$. A \textbf{$\boldsymbol{k}$-leaf} is a set of $k$ connected vertices such that $k-1$ have degree 2 and one has degree 1. A \textbf{$k$-leaf-pair} consists of two $k$-leaves joined together at a root vertex (see e.g. Figure \ref{Fig: Leaves})
\end{dfn}
In other words, a $k$-leaf is a line graph of length $k$ in which one end is connected to some base graph.
	
\begin{figure}[t]
\centerline{\includegraphics[width=0.45\textwidth]{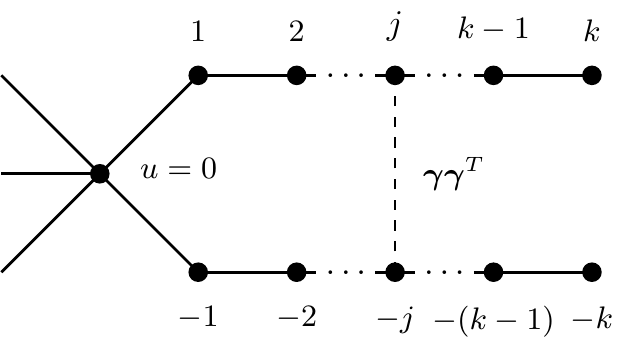}}
\caption{A graph $G$ containing a $k$-leaf-pair connected to some root vertex $u=0$, the graph $\bar{G}$ is formed by inserting the edge $(j,-j)$ which can be achieved via the rank one matrix $[\vec{\gamma} \vec{\gamma}^T]_{ik} = (\delta_{i,j} - \delta_{i,-j})(\delta_{k,j} - \delta_{k,-j})$.} \label{Fig: Leaves}
\end{figure}

The basic premise of this mechanism will be that if we have $k$-leaf-pair connected to the rest of the graph at some root vertex, $u$ say, then there exists eigenvectors which are only non-zero on these leaves. If an edge is inserted appropriately, only these eigenvectors will change, whereas the rest will be left unaltered with the same eigenvalues.

\begin{lem}\label{Lem: Leaf insertion}
Let $G_1$ and $G_2$ be two graphs, which each contain a $k$-leaf-pair. Let $\bar{G}_1$ and $\bar{G}_2$ be the respective graphs obtained by inserting an edge connecting the $j^{\rm th}$ vertex from the root in each $k$-leaf. Then $\sigma(G_1) = \sigma(G_2) \Leftrightarrow \sigma(\bar{G}_1) = \sigma(\bar{G}_2)$.
\end{lem}

\begin{proof}For convenience, let us number the vertices in $G_1$ and $\bar{G}_1$ in the following manner $\V := \{-k,-(k-1),\ldots,-1,0,1,\ldots,V - k-1\}$ and denote the subset $\V_{\rm leaves} := \{-k,\ldots,k\}$ such that the two $k$-leaf ends are $k$ and $-k$ and are joined to the rest of the graph at a root vertex $u=0$, as illustrated in Figure \ref{Fig: Leaves}. Suppose we insert the edge at the point $(j,-j)$ for some $j \in \V_{\rm leaves}$ then we can introduce the vector $\vec{\gamma}$ which takes the values $\gamma(i) = \delta_{i,j} - \delta_{i,-j}$,
\[
L(\bar{G}_1) = L(G_1) + \vec{\gamma} \vec{\gamma}^T.
\]
Now, let $\pi : \R^V \to \R^V$ be the operator that exchanges the two leaves, whilst the rest of the graph remains invariant, i.e.
\begin{equation}\label{Eqn: Symmetry action}
[\pi \f](i) = 
\Bigg\{\begin{array}{cl}
f(-i) & i \in \V_{{\rm leaves}} \\
f(i) & i \in \V \setminus  \V_{{\rm leaves}}.
\end{array} 
\end{equation}
Since $\pi^2 = I_V$ we can decompose our space into those vectors which are even $(+)$ and odd $(-)$ under this transformation,
\begin{equation}\label{Eqn: Symmetry decomposition}
\calS_{\pm} := \{\f : \pi \f = \pm \f \}
\end{equation}
and $\R^V \cong \calS_{+} \oplus \calS_{-}$. Furthermore, if $\f \in \calS_{+}$ then $\langle \vec{\gamma},\f\rangle := \vec{\gamma}^T\f = 0$ and thus, if it is an eigenvector of $L(G_1)$, we have
\begin{equation}\label{Eq: Same eigenvalue}
L(\bar{G}_1)\f = L(G_1)\f + \vec{\gamma} \vec{\gamma}^T\f = \lambda\f,
\end{equation}
meaning it is also an eigenvector of $L(\bar{G}_1)$ with the same eigenvalue.

Now, since $[\pi, L(G_1)] = [\pi,L(\bar{G}_1)] = 0$ (both Laplacians are symmetric under the action of $\pi$) we can choose a basis in which all eigenvectors $\f_n$  are either even or odd. The corresponding spectra we denote by $\sigma^{\pm}(G_1)$ such that $\sigma(G_1) = \sigma^{+}(G_1) \cup \sigma^{-}(G_1)$ and similarly for $\bar{G}_1$. All the above holds similarly for $L(G_2)$ and $L(\bar{G}_2)$. Hence, since $\sigma^{+}(G_1) = \sigma^{+}(\bar{G}_1)$ and $\sigma^{+}(G_2) = \sigma^{+}(\bar{G}_2)$, we have
\begin{equation}
\sigma^+(G_1) = \sigma^+(G_2) \Leftrightarrow \sigma^+(\bar{G}_1) = \sigma^+(\bar{G}_2),
\end{equation}

It thus remains to show that $\sigma^-(G_1) = \sigma^-(G_2) \Leftrightarrow \sigma^-(\bar{G}_1) = \sigma^-(\bar{G}_2)$. If $\f \in \calS_{-}$ then the action (\ref{Eqn: Symmetry action}) and decomposition (\ref{Eqn: Symmetry decomposition}) imply that for all $i \in \{0\} \cup  (\V \setminus  \V_{\rm leaves})$ we have $[\pi f](i) = f(i) = -f(i)$, i.e. $f(i)=0$. Therefore if $\f \in \calS_{-}$ (it is only non-zero on $\V_{\rm leaves}$) is an eigenvector of $L(G_1)$ with eigenvalue $\lambda$ it is also an eigenvector of $L(G_2)$ with the same eigenvalue and hence $\sigma^{-}(G_1) = \sigma^{-}(G_2)$. The same reasoning gives us that $\sigma^{-}(\bar{G}_1) = \sigma^{-}(\bar{G}_2)$, which completes the result. However note that $\sigma^{-}(G_1) \neq \sigma^{-}(\bar{G}_1)$, in contrast to the even part of the spectrum.
\end{proof}

\begin{thm}\label{Thm: Leaves flip count}
Let $G_1$ and $G_2$ be two graphs satisfying $\sigma(G_1) = \sigma(G_2)$ in which each graph contains a $k$-leaf-pair. Let us denote the respective eigenvectors of $G_1$ and $G_2$ by $\f_n$ and $\g_n$ and suppose that $\mu_{G_1}(\f_n) = \mu_{G_2}(\g_n)$ and $\nu_{G_1}(\f_n)=\nu_{G_2}(\g_n)$ for all generic $n$. If $\bar{G}_1$ and $\bar{G}_2$ are the graphs,  with associated eigenvectors $\bar{\f}_n$ and $\bar{\g}_n$, obtained by inserting an edge between the two corresponding vertices in each $k$-leaf. Then
\begin{enumerate}
\item [(i)] $\sigma(\bar{G}_1) = \sigma(\bar{G}_2)$ (they are isospectral)
\item [(ii)] For all generic eigenvectors $\mu_{\bar{G}_1}(\bar{\f}_n) = \mu_{\bar{G}_2}(\bar{\g}_n)$ (they have the same flip count)
\item [(iii)] For all generic eigenvectors $\nu_{\bar{G}_1}(\bar{\f}_n) = \nu_{\bar{G}_2}(\bar{\g}_n)$ (they have the same nodal count)
\end{enumerate}
\end{thm}
\begin{proof}Part (i) is simply restating the result of Lemma \ref{Lem: Leaf insertion}. 

To show Part (ii), using the reasoning from the proof of Lemma \ref{Lem: Leaf insertion}, we know that if $\f$ is a generic eigenvector of $L(G_1)$ with eigenvalue $\lambda$ then it must belong to $\f \in \calS_{+}$ and, hence, is also an eigenvector of $L(\bar{G}_1)$ with the same eigenvalue. In addition, if $(j,-j)=\E(\bar{G}_1) \setminus \E(G_1)$ is the edge inserted to make $\bar{G}_1$ then $f(j) = f(-j)$. Therefore we have $\mu_{G_1}(\f) = \mu_{\bar{G}_1}(\f)$ since $\f$ does not change sign across the edge $(j,-j)$. The same argument applies to give $\mu_{G_2}(\g) = \mu_{\bar{G}_2}(\g)$. Finally, by the isospectrality result of Lemma \ref{Lem: Leaf insertion}, we know that if $\f$ and $\g$ are the $m^{\rm th}$ eigenvectors in $G_1$ and $G_2$ then they must occupy the same position in the spectrum, say the $n^{\rm th}$, in $\bar{G}_1$ and $\bar{G}_2$ and so $\mu_{\bar{G}_1}(\bar{\f}_n) = \mu_{G_1}(\f_m) = \mu_{G_2}(\g_m) = \mu_{\bar{G}_2}(\bar{\g}_n)$ for all generic $n$.

In order to prove Part (iii) we require more information about the eigenvectors than their value at the vertices $j$ and $-j$. If $k$ denotes the endpoint of the $k$-leaf, then by the eigenvalue equation $L(G_1) \f = \lambda \f$ we have
\begin{equation}\label{Eqn: Eigenvalue eqn leaf}
\lambda f(k) = f(k) - f(k-1) \ \  \Rightarrow \ \ f(k) = \frac{1}{(1- \lambda)} f(k-1).
\end{equation}
Thus $f(k)$ depends only on the eigenvalue $\lambda$ and the value of $\f$ at the vertex preceding it. Similarly $\lambda f(k-1) = 2f(k-1) - f(k) - f(k-2)$, which, using (\ref{Eqn: Eigenvalue eqn leaf}) gives
\[
f(k-1)  = \left((2-\lambda) - \frac{1}{(1-\lambda)}\right)^{-1}f(k-2),
\]
so, again, we see that $f(k-1)$ can be constructed by simply knowing $\lambda$ and the preceding value $f(k-2)$. The recursive nature therefore implies that on the leaves the vector takes the form $f(i) = F_i(\lambda)f(0)$ for some function $F_i(\lambda)$ (it's exact nature is not important here, only that it depends on $\lambda$).

By inserting the edge $(j,-j)$ we therefore find $\nu_{\bar{G}_1}(\f) = \nu_{G_1}(\f) + \chi(\lambda,j)$, where $\chi(\lambda,j)$ (which only depends on $\lambda$ and the inserted edge $(j,-j)$) is either 0 (the vertices $j$ and $-j$ belong to the same nodal domain of $\f$), or $-1$ (the vertices $j$ and $-j$ belong to different nodal domains of $\f$). Applying the same arguments to a generic eigenvector $\g$ of $G_2$ with the same eigenvalue $\lambda$ gives $\nu_{\bar{G}_2}(\f) = \nu_{G_2}(\f) + \chi(\lambda,j)$ and so, by the isospectrality condition and the equality of nodal domains in $G_1$ and $G_2$ we find $\nu_{\bar{G}_2}(\bar{\f}_n) = \nu_{G_1}(\f_m) + \chi(\lambda,j) = \nu_{G_2}(\g_m) + \chi(\lambda,j) = \nu_{\bar{G}_2}(\bar{\g}_n)$ for all generic $n$.
\end{proof}

The main issue to highlight is that, although all the generic eigenvectors themselves do not change between $G$ and $\bar{G}$, some of their positions in the spectrum will. Therefore the isospectrality condition ensures that they will change in the same manner in both $G$ and $\bar{G}$, ensuring the flip count $\mu_n$ of the associated $n^{{\rm th}}$ eigenvectors is the same.

Importantly, we do not preclude in Theorem \ref{Thm: Leaves flip count} that the graphs $G_1$ are $G_2$ are isomorphic. In fact the easiest way to obtain examples to illustrate this theorem comes from taking $G_1 = G_2$ but containing multiple sets of $k$-leaf pairs. Adding edges to different pairs then makes it possible to create a $\bar{G}_1$ and $\bar{G}_2$ that are non-isomorphic, but are still isospectral, which is how the isospectral pair in Figure \ref{fig:overlap} are created.

\begin{figure}[t]
\centering
\includegraphics[width=0.75\textwidth]{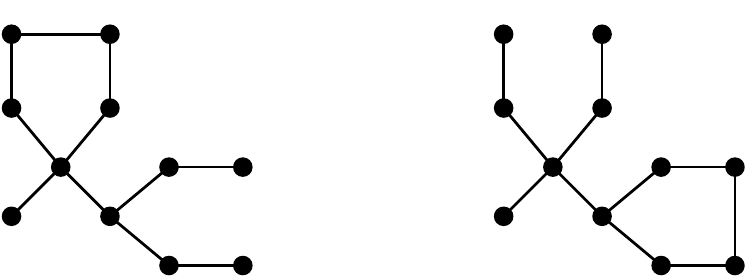}
\caption{A pair of isospectral graphs with the same flip count and nodal count which originate by adding edges to trees.}
\label{fig:overlap}
\end{figure}

\subsection{Removing the isospectrality condition}\label{Sec: Removing isospectrality}

So far we have presented examples of mechanisms which are able to create families of isospectral graphs with the same flip count and nodal count. One may ask whether this condition is always needed. In the following we show this is not the case and present a way of constructing non-isospectral graphs with the same flip count and also the same nodal count. This is born out of Lemma \ref{Lem: Leaf insertion} and the observation that all trees have a flip count $\mu_n = n-1$ and nodal count $\nu_n = n$.

\begin{thm}\label{thm: Edge insertion 2}
Let $G$ be a graph containing a $k$-leaf-pair of length $k=1$ and $\bar{G}$ the graph obtained by inserting an edge between the ends of these leaves. If we denote $\f_n$ to be the $n^{\rm th}$ eigenvector of $G$ then the flip count of the $n^{\rm th}$ eigenvector $\bar{\f}_n$ of $\bar{G}$ is given (provided $\f_n$ is generic) by
\begin{equation}\label{Leaf flip count}
\mu_{\bar{G}}(\bar{\f}_{n}) = \left\{ \begin{array}{ll}
\mu_{G}(\f_{n+1}) &  {\rm if} \;  1<\lambda_n<3 \\
\mu_{G}(\f_{n}) & {\rm otherwise}
\end{array}\right.
\end{equation}
and the nodal count by
\begin{equation}\label{Leaf nodal count}
\nu_{\bar{G}}(\bar{\f}_n) = \left\{ \begin{array}{ll}
\nu_G(\f_n) & {\rm if} \; \lambda_n <1 \\
\nu_G(\f_{n+1}) -1 &  {\rm if} \;  1<\lambda_n<3 \\
\nu_G(\f_n) -1 & {\rm if} \;  \lambda_n > 3 \ .
\end{array}\right.
\end{equation}
\end{thm}

\begin{proof}
The statement (\ref{Leaf flip count}) follows almost immediately from Lemma \ref{Lem: Leaf insertion}. In this case our leaves are of length 1 and so the subspace $\calS_{-}$ is spanned entirely by $\vec{\gamma}$ (given by $\gamma(i) = \delta_{i,1} - \delta_{i,-1}$). Thus, since $\vec{\gamma}$ is a non generic eigenvector of $L(G)$ with eigenvalue 1 we have
\[
L(\bar{G})\vec{\gamma} = L(G)\vec{\gamma} + \vec{\gamma}\Iprod{\vec{\gamma}}{\vec{\gamma}} = 3\vec{\gamma}.
\]
The rank 1 perturbation $\vec{\gamma}$ induces a shift of one eigenvalue from 1 to 3, while all other eigenvalues remain the same. Therefore, by Lemma \ref{Lem: Leaf insertion} we have for all eigenvectors, except $\vec{\gamma}$,
\begin{equation}\label{Eqn: Eigenvector relation}
\bar{\f}_n = \left\{ \begin{array}{ll}
\f_{n+1} &  {\rm if} \;  1<\lambda_n<3 \\
\f_n & {\rm otherwise}
\end{array}\right.
\end{equation}
and so (\ref{Leaf flip count}) follows by the fact that $\Iprod{\vec{\gamma}}{\vec{\bar{f}}}=0 \Rightarrow \bar{f}(1) = \bar{f}(-1)$, i.e. $\bar{f}$ does not change sign across the inserted edge $(1,-1)$.

To establish (\ref{Leaf nodal count}) requires some more information, specifically the sign of $f(0)$ in comparison to $f(1) = f(-1)$. Evaluating the Laplacian of $G$ at the vertex $1$ we get
\[
[L(G)\vec{f}](1) = \lambda f(1) = f(1) - f(0) \Rightarrow (1- \lambda) = \frac{f(0)}{f(1)}.
\]
Hence, for $\lambda <1$, $\sgn(f(0)) = \sgn(f(1)) = \sgn(f(-1))$ and inserting an edge does not increase the number of nodal domains, whereas for $\lambda > 1$ the opposite is true. Combining this with (\ref{Eqn: Eigenvector relation}) leads to (\ref{Leaf nodal count}).
\end{proof}

As a consequence of Theorem \ref{thm: Edge insertion 2} we are able to get the following corollary, which enables us to generate pairs of non-isospectral graphs with the same flip count

\begin{cor}\label{Cor: Non_cospec}
Let $G_1$ and $G_2$ be graphs each containing a $k$-leaf-pair of length $k=1$ with eigenvectors $\f_n$ and $\g_n$ respectively. Suppose that for all generic eigenvectors we have $\mu_{G_1}(\f_n) = \mu_{G_2}(\g_n)$ and $\nu_{G_1}(\f_n) = \nu_{G_2}(\g_n)$ and they have the same number of eigenvalues in the ranges $0 \leq \lambda < 1$, $1 \leq\lambda < 3$ and $\lambda \geq 3$. Then, if $\bar{G}_1$ and $\bar{G}_2$ are the graphs obtained by inserting an edge between each leaf, then $\mu_{\bar{G}_1}(\bar{\f}_n) = \mu_{\bar{G}_2}(\bar{\g}_n)$ and $\nu_{\bar{G}_1}(\bar{\f}_n) = \nu_{\bar{G}_2}(\bar{\g}_n)$ for all generic $n$.
\end{cor}

\begin{figure}[t]
\centerline
{\includegraphics[width=0.9\textwidth]{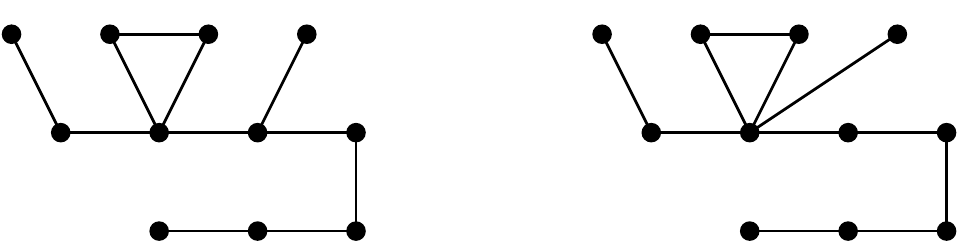}}
\vspace{10pt}
\centerline{$\sigma(\bar{G}_1) = \{\framebox{0},\framebox{0.13},\framebox{0.42},\framebox{0.61},1,$ \hspace{70pt} $\sigma(\bar{G}_2) = \{\framebox{0},\framebox{0.13},\framebox{0.48},\framebox{0.72},1,$}
\vspace{10pt}
\centerline{\small \hspace{40pt} $ \framebox{2.09},\framebox{2.39},3,\framebox{3.21},\framebox{3.81},\framebox{5.34}\}$ \hspace{80pt} $\framebox{1.67},\framebox{2.46},\framebox{2.80},3,\framebox{3.66},\framebox{6.10}\}$}
\vspace{5pt}
\caption{An example of two non-isospectral graphs in which all generic eigenvectors have the same flip count, given by $\mu_{\bar{G}_1} = \mu_{\bar{G}_2} =\{ 0   ,  1   ,  2  ,   3  ,      6  ,   7  ,      8  ,   9  ,  10 \}$ and nodal count $\nu_{\bar{G}_1} = \nu_{\bar{G}_2}= \{1   ,  2   ,  3 ,    4   ,  6   ,  7   ,  8   ,  9   , 10\}$. The spectrum is presented under each graph and eigenvalues associated to generic eigenvectors are marked in boxes.}
\label{Fig: Non_cospec}
\end{figure}

Corollary \ref{Cor: Non_cospec} states that, via this mechanism, we are required to find non-isospectral graphs $G_1$ and $G_2$ with the same flip count in order to generate another pair with of non-isospectral graphs $\bar{G}_1$ and $\bar{G}_2$ with the same flip count. At first sight this may seem like a pointless search, however it turn out to be very fruitful, as it allows us to go from tree-graphs (with trivial topologies) to non-tree-graphs, which no longer have trivial topologies. Examples of graphs $\bar{G}_1$ and $\bar{G}_2$ generated in this way are given in Figure \ref{Fig: Non_cospec}.

\section{Quantum graphs}\label{Sec: Quantum graphs}

\begin{dfn}In the case of a metric graph we define an \textbf{$l$-leaf} to be a leaf with edge length $l$. In addition, we define an \textbf{$l$-leaf-pair} to be two $l$-leaves connected via some root vertex.
\end{dfn}

\begin{remark}On a quantum graph, a vertex of degree two with Neumann vertex conditions (a so-called dummy vertex) does not alter any of the spectral properties \cite{Berk-Kuch-2013}. Therefore we are free to add such vertices on leaves at any position we desire (see e.g. Figure \ref{Fig: Q graph construct}).
\end{remark}

In the previous section we introduced the mechanism of inserting an edge between a pair of leaves. Our aim is to introduce a similar mechanism for quantum graphs, however the process of inserting an edge would increase the overall length of the graph, which could have dramatic consequences for the spectrum. The solution we find is thus to glue the leaves together at the same point, which alters the topology of the graph but retains the same overall length. This process leads us to Theorem \ref{Thm: Q graphs main thm} below, which is the counterpart to Theorem \ref{Thm: Leaves flip count} in Section \ref{Sec: Discrete graphs}. However, to begin, we start with the following lemma.

\begin{figure}[t]
\centerline
{\includegraphics[width=.9\linewidth]{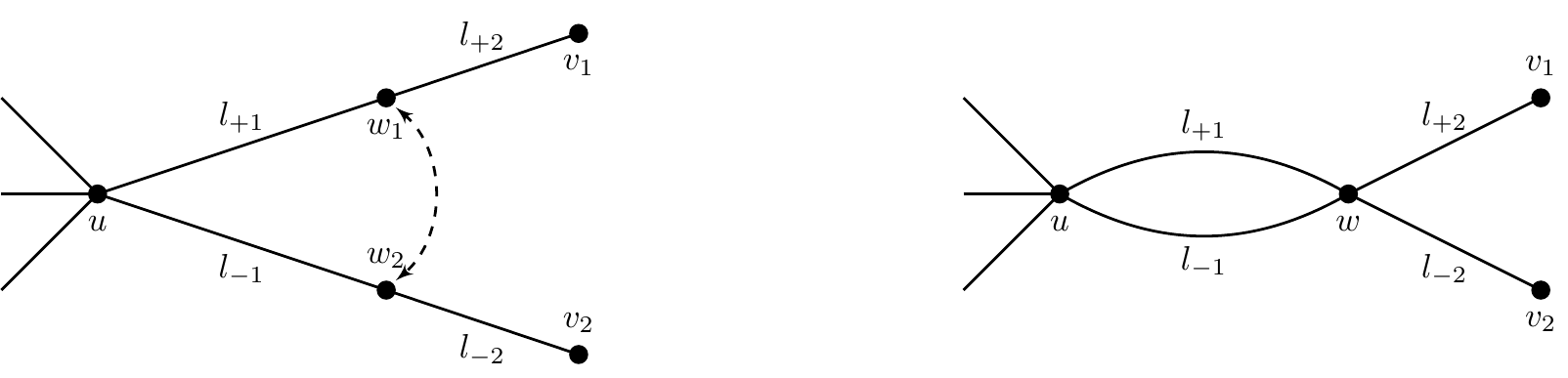}}
\caption{The process of gluing the leaves of a quantum graph $\Gamma$ together at the vertices $w_1$ and $w_2$ to form a new quantum graph $\bar{\Gamma}$ with vertex $w$ of degree 4}
\label{Fig: Q graph construct}
\end{figure}

\begin{lem}\label{Lem: Q graph perturbation} Let $\Gamma_1$ and $\Gamma_2$ be two graphs each containing a $l$-leaf-pair. Let $\bar{\Gamma}_1$ and $\bar{\Gamma}_2$ be the respective graphs obtained by gluing the leaves in each pair together at the same point. Then $\sigma(\Gamma_1) = \sigma(\Gamma_2) \iff \sigma(\bar{\Gamma}_1) = \sigma(\bar{\Gamma}_2) $.
\end{lem}

\begin{proof}The proof proceeds in analogous manner to the discrete case given in Lemma \ref{Lem: Leaf insertion}. Let us first administer some notation. In $\Gamma_1$ let $u$ denote the root vertex which joins the two leaves, with $v_{\pm}$ the ends of these leaves. Let us introduce two dummy vertices $w_{\pm}$ the same distance away from the root on these respective leaves. Thus we have four edges which we denote $e_{\pm 1} = (w_{\pm},u)$ and $e_{\pm 2} = (w_{\pm},v_{\pm})$, with lengths $l_{e_{+1}} = l_{e_{-1}} = l_1$ and $l_{e_{+2}} = l_{e_{-2}} = l_2$, as illustrated in Figure \ref{Fig: Q graph construct}. Thus our edge set consists of the edges $\E = \{e_{-2},e_{-1},e_{1},e_2,e_3,\ldots,e_{|\E| - 2}\}$ and $\E_{\rm leaves} := \{e_{-2},e_{-1},e_{1},e_2\}$. Note that here $|\E|$ is the number of edges after we have inserted the dummy vertices $w_{\pm}$.

We now introduce the operator $\pi : L^2(\Gamma_1) \to L^2(\Gamma_1)$ that interchanges the two edges, i.e. for functions $\psi \in L^2(\Gamma_1)$ we have
\begin{equation}\label{Eqn: QG Symmetry action}
[\pi \psi](x_{e_i}) = \Bigg\{ \begin{array}{cl}
\psi(x_{e_{-i}}) & i = -2,-1,1,2 \\
\psi(x_{e_{i}}) & i = 3,\ldots, |\E| - 2.
\end{array}
\end{equation}
Again, since $\pi^2 = I$, the identity, we can decompose our Hilbert space into those functions which are even $(+)$ and odd $(-)$ under exchange of the two leaves
\begin{equation}\label{Eqn: Hilbert space decomposition}
L^2_{\pm}(\Gamma_1) : = \{ \psi \in L^2(\Gamma_1) : \pi \psi = \pm \psi \},
\end{equation}
with $L^2(\Gamma_1) = L^2_+(\Gamma_1) \oplus L^2_-(\Gamma_1)$. Moreover, if $\psi \in L^2_+(\Gamma_1)$ and is an eigenvector of $L_{\Gamma_1}$ with eigenvalue $\lambda$ then it is also an eigenvector of $L_{\bar{\Gamma}_1}$ with the same eigenvalue. To see this note that if $L_{\Gamma_1} \psi = \lambda\psi$  then $-d^2 \psi (x_e) /d x_e^2 = \lambda \psi(x_e)$ on all edges $e \in \E(\bar{\Gamma}_1)$ and thus it remains to check the vertex conditions at $w_{\pm}$ are satisfied.

\begin{enumerate}
\item \textbf{Continuity:} By the vertex continuity at $w_{\pm}$ on $\Gamma_1$ and the fact that $\psi$ is even we have $\psi_{e_1}(0) = \psi_{e_2}(0) = \psi_{e_{-2}}(0) = \psi_{e_{-1}}(0)$, which is precisely the condition required at $w$ on $\bar{\Gamma}_1$.
\item \textbf{Derivatives:} On $\Gamma_1$ the derivatives at $w_{\pm}$ satisfy $\psi'_{e_1}(0) +  \psi'_{e_2}(0) = 0$ and $\psi'_{e_{-1}}(0) + \psi'_{e_{-2}}(0) = 0$ respectively. Adding these relations together gives $\psi'_{e_1}(0) + \psi'_{e_2}(0) + \psi'_{e_{-1}}(0) + \psi'_{e_{-2}}(0) = 0$, which again is the condition required for the derivatives of $\psi$ at $w$ on $\bar{\Gamma}_1$.
\end{enumerate}
In addition, since $\pi$ preserves the eigenspaces of $\Gamma_1$ we can choose a basis of eigenfunctions $\psi_n$ such that they are either even or odd and we denote the corresponding spectra (including multiplicities) by $\sigma^{\pm}(\Gamma_1)$, so that $\sigma(\Gamma_1) = \sigma^+(\Gamma_1) \cup \sigma^-(\Gamma_1)$ and similarly for $\bar{\Gamma}_1$.

The same result hold for $\Gamma_2$ and $\bar{\Gamma}_2$. Hence, since $\sigma^+(\Gamma_1) = \sigma^+(\bar{\Gamma}_1)$ and $\sigma^+(\Gamma_2) = \sigma^+(\bar{\Gamma}_2)$, we have
\[
\sigma^+(\Gamma_1) = \sigma^+(\Gamma_2) \iff \sigma^+(\bar{\Gamma}_1) = \sigma^+(\bar{\Gamma}_2).
\]
It thus remains to show the same for the odd subspaces. We note that if $\psi \in L_{-}(\Gamma_1)$ then (\ref{Eqn: QG Symmetry action}) and (\ref{Eqn: Hilbert space decomposition}) imply that $\psi(x) = 0$ for all points $x$ which are not on the leaves of $\Gamma_1$ and similarly for $\Gamma_2$. Therefore, $\psi \in L_{-}(\Gamma_1)$ is an eigenfunction of $\Gamma_1$ with eigenvalue $\lambda$ if and only if it is an eigenfunction of $\Gamma_2$ with the same eigenvalue and so $\sigma^-(\Gamma_1) = \sigma^-(\Gamma_2)$. The same reasoning implies that $\sigma^-(\bar{\Gamma}_1) = \sigma^-(\bar{\Gamma}_2)$, which completes the result.
\end{proof}

\begin{thm}\label{Thm: Q graphs main thm}Let $\Gamma_1$ and $\Gamma_2$ be two graphs each containing an $l$-leaf-pair and satisfying $\sigma(\Gamma_1) = \sigma(\Gamma_2)$. Let us denote the respective eigenfunctions by $\psi_n$ and $\phi_n$ and suppose that $\mu_{\Gamma_1}(\psi_n) = \mu_{\Gamma_1}(\phi_n)$ and $\nu_{\Gamma_1}(\psi_n)=\nu_{\Gamma_2}(\phi_n)$ for all generic $n$. Let also $\bar{\Gamma}_1$ and $\bar{\Gamma}_2$ be the respective graphs obtained by gluing the leaves together at points $w_{\pm}$ in the fashion of Lemma \ref{Lem: Q graph perturbation} and $\bar{\psi}_n$ and $\bar{\phi}_n$ be the respective eigenfunctions. Then 
\begin{enumerate}
\item [(i)]$\sigma(\bar{\Gamma}_1) = \sigma(\bar{\Gamma}_2)$ (the graphs are isospectral).
\item [(ii)]For all generic eigenvectors $\mu_{\bar{\Gamma}_1}(\bar{\psi}_n) = \mu_{\bar{\Gamma}_2}(\bar{\phi}_n)$ (they have the same flip count).
\item [(iii)]For all generic eigenvectors $\nu_{\bar{\Gamma}_1}(\bar{\psi}_n) = \nu_{\bar{\Gamma}_2}(\bar{\phi}_n)$ (they have the same nodal count).
\end{enumerate}
\end{thm}

\begin{remark}Depending on the choice of $w_{\pm}$ it may be the case that eigenfunctions that are generic in $\Gamma_1$ (without the dummy vertices inserted) are no longer generic in $\bar{\Gamma}_1$, since they may be zero at $w_{\pm}$. However the combined zero set of all generic eigenfunctions forms a countable sequence of points on the edges of the leaves. Therefore, as this set is of zero measure, we can choose $w_{\pm}$ almost everywhere on the leaves such that all generic eigenfunctions of $\Gamma_1$ are also generic on $\bar{\Gamma}_1$. The same holds for $\Gamma_2$ and $\bar{\Gamma}_2$.
\end{remark}

\begin{proof}[Proof of Theorem \ref{Thm: Q graphs main thm}] Part (i) is simply restating Lemma \ref{Lem: Q graph perturbation}.

For Part (ii) we note that if $\psi$ is a generic eigenfunction of $\Gamma_1$ with eigenvalue $\lambda$, then $\psi \in L^2_+(\Gamma_1)$ and hence, by the argument in the proof of Lemma \ref{Lem: Q graph perturbation} it is an eigenfunction of $\bar{\Gamma}_1$ with the same eigenvalue. If $\psi$ is the $m^{{\rm th}}$ eigenfunction $\psi_m$ of $\Gamma_1$ and the $n^{{\rm th}}$ eigenfunction $\bar{\psi}_n$ of $\bar{\Gamma}_1$ then $\mu_{\Gamma_1}(\psi_m) = \mu_{\bar{\Gamma}_1}(\bar{\psi}_n)$, since the process of gluing the edges together does not induce any more zeros. For the same reason $\mu_{\Gamma_2}(\phi_r) = \mu_{\bar{\Gamma}_2}(\bar{\phi}_n)$ when $\phi$ is generic and the $n^{{\rm th}}$ (resp. $r^{{\rm th}}$) eigenfunction of $\Gamma_2$ (resp. $\bar{\Gamma}_2$). Thus Lemma \ref{Lem: Q graph perturbation} ensures that, since $\bar{\Gamma}_1$ and $\bar{\Gamma}_2$ are isospectral, we have $m = r$ and thus $\mu_{\bar{\Gamma}_1}(\bar{\psi}_n) =  \mu_{\bar{\Gamma}_2}(\bar{\phi}_n)$ for all generic $n$.

For Part (iii) we need to deduce some more information about the eigenfunctions. Firstly, if $\psi$ is an eigenfunction of $\Gamma_1$ we know that on each edge it must be of the form
\[
\psi(x) = A_1\sin(\sqrt{\lambda} x) + B_1\cos(\sqrt{\lambda} x).
\]
Now, let $x \in [0,l]$ denote the position on one of the leaves with the point $x=0$ giving the position at the end of the leaf and $x=l$ the point of the root vertex $u$. Since all the vertices obey Neumann conditions we find that
\[
\frac{d \psi}{dx}(0) = A_1\sqrt{\lambda}\cos(0) - B_1\sqrt{\lambda}\sin(0) = A_1\sqrt{\lambda} = 0.
\]
Thus $\psi(x) = B_1\cos(\sqrt{\lambda}x)$. For the same reasoning we have on the graph $\Gamma_2$ $\phi(x) = B_2\cos(\sqrt{\lambda}x)$, for some other constant $B_2$. Gluing the edges together at the points $w_{\pm}$ we either have that the number of nodal domains is the same $\nu_{\Gamma_1}(\psi) = \nu_{\bar{\Gamma}_1}(\psi)$ (the points $w_{\pm}$ belong to the same nodal domain of $\psi$ on $\Gamma_1$) or it decreases by one, i.e. $\nu_{\bar{\Gamma}_1}(\psi)= \nu_{\Gamma_1}(\psi) -1$ (the points $w_{\pm}$ belong to different nodal domains of $\psi$ on $\Gamma_1$). However the same must be true for $\phi$ since $\psi(x)/\phi(x)$ is constant for all $x \in \E_{\rm leaves}$. Therefore $\nu_{\bar{\Gamma}_1}(\bar{\psi}_n) =  \nu_{\bar{\Gamma}_2}(\bar{\phi}_n)$ and the isospectrality of Lemma \ref{Lem: Q graph perturbation} ensures they must occupy the same position in the spectrum.
\end{proof}

We also comment that one could, in principle, obtain a pair of non-isospectral quantum graphs with the same flip and nodal counts in an analogous manner to Section \ref{Sec: Removing isospectrality}. The process would involve taking two tree graphs $\Gamma_1$ and $\Gamma_2$, both containing an $l$-leaf-pair, and then gluing together the leaves to obtain $\bar{\Gamma}_1$ and $\bar{\Gamma}_2$, as outlined above. $\Gamma_1$ will have the same eigenvalues in the odd part of the spectrum as $\Gamma_2$ (i.e. $\lambda^{{\rm odd}}_n(\Gamma_1) = \lambda^{{\rm odd}}_n(\Gamma_2)$) and, similarly, $\bar{\Gamma}_1$ the same (odd) eigenvalues as $\bar{\Gamma}_2$ (i.e. $\lambda^{{\rm odd}}_n(\bar{\Gamma}_1) = \lambda^{{\rm odd}}_n(\bar{\Gamma}_2)$). 

These odd eigenvalues can be explicitly calculated. They correspond to the eigenvalues of a line graph with Dirichlet conditions at one end, Neumann at the other and Neumann (resp. Dirichlet) for $\Gamma$ (resp. $\bar{\Gamma}$) at a central vertex a distance $l_1$ from the end with Dirichlet conditions. This means the associated eigenvalues are strictly interlacing, i.e. $\lambda^{{\rm odd}}_n(\Gamma_1) < \lambda^{{\rm odd}}_n(\bar{\Gamma}_1) < \lambda^{{\rm odd}}_{n+1}(\Gamma_1)$ (see e.g. Theorem 3.1.8 in \cite{Berk-Kuch-2013}).

Since $\Gamma_1$ and $\Gamma_2$ are trees they have the same flip and nodal counts. The even eigenvalues of $\Gamma_1$ are the same as $\bar{\Gamma}_1$ (similarly for $\Gamma_2$ and $\bar{\Gamma}_2$). Thus, if the positions in the spectrum (not necessarily the values) of these eigenvalues are the same - i.e. $\Gamma_1$ and $\Gamma_2$ have the same number of \emph{even} eigenvalues in each interval $(\lambda^{{\rm odd}}_1,\bar{\lambda}^{{\rm odd}}_1)$, $(\bar{\lambda}^{{\rm odd}}_1,\lambda^{{\rm odd}}_2)$ etc. - then the flip counts (although not necessarily the nodal counts) of $\bar{\Gamma}_1$ and $\bar{\Gamma}_2$ will coincide. However, due to the fact that there are an infinite number of eigenvalues for quantum graphs, this cannot be checked numerically. Thus, at present, we are unsure about the existence of such a pair of quantum graphs.

\section{Conclusions}\label{Sec: Conclusions}
The work presented here outlines a simple mechanism for constructing both discrete and quantum isospectral graphs with the same flip and nodal count. Thus presenting a family of counter-examples to the conjecture that pairs of isospectral (non-tree) graphs may be distinguished by their nodal properties  \cite{band_nodal_2006,Band-2007,oren_isospectral_2012}. In addition, we have also shown that certain non-isospectral graphs share the same flip count and nodal count and thus one cannot deduce that two graphs are non-isospectral simply by knowing the nodal properties.

We believe the work highlights a number of previous questions and raises yet more. For instance, what is the likelihood of finding isospectral graphs and quantum graphs with the same flip/nodal count? - Numerical simulations by I. Oren suggest these may be extremely rare in the cases such as random regular graphs \cite{idan_oren_private}. With our methods, we can create graphs with very different topology that are isospectral and have identical flip counts and nodal counts. What properties unite all of these examples? and is there a criterion we can point to that tells us whether isospectral graphs and quantum graphs will have the same flip/nodal count or are these examples all coincidental? Perhaps recent advances in our understanding of nodal and flip counts in connection with the stability of eigenvalues \cite{berkolaiko_stability_2012,Band-2012,berkolaiko_nodal_2013,colin_de_verdiere_magnetic_2013} may provide answers, as this makes connections with topological properties of the graph. Progress in this respect would be greatly applicable to graph-theoretical problems looking at distinguishing non-isomorphic graphs and may also indicate how to proceed in the case of bounded domains.

\subsection*{Acknowledgements} We would like to extend our gratitudes to K. Ammann for discussions regarding her unpublished work, to U. Smilansky for his guidance and support during this project and to L. Alon and R. Band for helpful comments on a preliminary form of this manuscript. We are also grateful for the financial help of the Feinberg Graduate School at the Weizmann Institute of Science, where much of this work was carried out. Finally, CHJ would like to acknowledge the Leverhulme Trust (ECF-2014-448) and JSJ would like to acknowledge the Danish Council for Independent Research for financial support.

\section*{References}

\bibliographystyle{ieeetr}
\bibliography{References}

\end{document}